\documentclass[pra,twocolumn,showpacs,superscriptaddress]{revtex4-1}

\usepackage[dvips]{graphicx}
\usepackage{amsmath,amssymb,amsthm,amscd}
\usepackage{bbold}
\usepackage{mdwlist}

\newcommand{\bra}[1]{\langle#1|}
\newcommand{\ket}[1]{|#1\rangle}
\newcommand{\braket}[2]{\langle#1|#2\rangle}

\newcommand{\Tr}{\operatorname{Tr}}

\newcommand{\pure}[2]{\left(\begin{array}{c}#1\\#2\end{array}\right)}

\newtheorem{dfn}{Definition}

\newtheorem{thm}{Proposition}
\newtheorem{cor}{Corollary}
\newtheorem{algo}{Algorithm}

\begin{document}

\title{Informational power of quantum measurements}

\author{Michele \surname{Dall'Arno}}
\author{Giacomo Mauro \surname{D'Ariano}}
\affiliation{Quit group, Dipartimento di Fisica ``A. Volta'', via A. Bassi 6,
  I-27100 Pavia, Italy}
\affiliation{Istituto Nazionale di Fisica Nucleare, Gruppo IV, via A. Bassi 6,
  I-27100 Pavia, Italy}

\author{Massimiliano F. \surname{Sacchi}}
\affiliation{Quit group, Dipartimento di Fisica ``A. Volta'', via A. Bassi 6,
  I-27100 Pavia, Italy}
\affiliation{Istituto di Fotonica e Nanotecnologie (INF-CNR), Piazza Leonardo
  da Vinci 32, I-20133, Milano, Italy}
	
\date{\today}

\begin{abstract} 
  We introduce the informational power of a quantum measurement as the maximum
  amount of classical information that the measurement can extract from any
  ensemble of quantum states. We prove the additivity by showing that the
  informational power corresponds to the classical capacity of a
  quantum-classical channel. We restate the problem of evaluating the
  informational power as the maximization of the accessible information of a
  suitable ensemble. We provide a numerical algorithm to find an optimal
  ensemble, and quantify the informational power.
\end{abstract}

\maketitle

\section{Introduction}

The information stored in a quantum system is accessible only through a
quantum measurement, and the postulates of quantum theory severely limit what
a measurement can achieve. The problem of evaluating the informational power
of a quantum measurement - i. e. how much informative the measurement is - has
not been addressed yet in the literature, despite the obvious practical
relevance in several contexts, such as the communication of classical
information over noisy quantum channels, the storage and retrieval of
information from quantum memories \cite{NC00}, and the purification of noisy
quantum measurements \cite{DDS10}.

For which ensemble of states a given quantum measurement is more informative?
To answer such question, one can consider two different figures of merit: the
probability of correct detection (in a discrimination scenario) and the mutual
information (in a communication scenario). Up to now, the only case of
optimization of the input ensemble in the literature \cite{EE07} considers the
former as a figure of merit, benefiting of its linearity that simplifies
calculations, and working out an explicit form for the optimal states and the
corresponding detection probability. The latter case of optimization, namely
the maximization of the mutual information over input ensembles, is the aim of
this work. To this purpose, we define the informational power as the maximal
mutual information that a given quantum measurement is able to extract from an
ensemble of quantum states. We call the optimal ensemble maximally
informative.

The problem has analogies with those of quantifying classical capacity of
quantum channels and of attaining accessible information \cite{NC00}. In fact,
as we will show, the informational power of a quantum measurement is the
channel capacity of a quantum-classical channel \cite{Hol98}, and the
evaluation of the informational power is the dual of the problem of accessible
information, in a sense that we will clarify later.

The paper is organized as follows. In Sect. \ref{sec:infopower} we introduce
the informational power of quantum measurements. We show that it is the
classical capacity of a quantum-classical channel and prove additivity. We
restate the problem of maximizing the informational power of a measurement as
the problem of maximizing the accessible information of a suitable ensemble,
and provide a bound on the minimal number of states of a maximally informative
ensemble. In Sect. \ref{sec:evaluation}, we provide a numerical algorithm to
find a maximally informative ensemble for a given quantum measurement. In
Sect. \ref{sec:examples}, we classify some quantum measurements according to
their informational power, namely quantum measurements with commuting
elements, real-symmetric and mirror-symmetric quantum measurements, and the
$2$-dimensional symmetric informationally complete quantum measurement (i. e.,
the tetrahedral measurement). We summarize our results in
Sect. \ref{sec:conclusions}.

\section{Informational power of quantum measurements}\label{sec:infopower}

Let us recall some basic definitions \cite{Cov06} and set the notation. A {\em
  random variable} $X=\{p_i,X_i\}$ is a set of outcomes $\{i\}$ with values
$\{X_i\}$ and prior probabilities $\{p_i\}$. A joint random variable
$(X^1,\dots X^N)$ is defined analogously.

A measure of the uncertainty associated with a random variable $X$ is given by
the {\em Shannon entropy} $H(X)$
\begin{equation}\label{eq:shentr}
  H(X) := -\sum_i p_i\log p_i,
\end{equation}
where $\log x$ denotes the logarithm to the base $2$. A measure of the
remaining uncertainty of a random variable Y given that the value of $X$ is
known is provided by the {\em conditional entropy} $H(Y|X)$
\begin{equation}
  H(Y|X) := H(X,Y) - H(X).
\end{equation}
A measure of how much two random variables $X$ and $Y$ are correlated is given
by the {\em mutual information}
\begin{equation}\label{eq:minfo}
  I(X:Y) := H(X) + H(Y) - H(X,Y).
\end{equation}
The expected value of the mutual information of two random variables $X$ and
$Y$, given the value of a third $Z$, is the {\em conditional mutual
  information}
\begin{equation}
  I(X:Y|Z) := H(Y|Z) - H(Y|X,Z).
\end{equation}
Given a Markov chain $X \to Y \to Z$, i.e. a set of three random variables
$X$, $Y$, and $Z$, with $Z$ conditionally independent of $X$, one has the {\em
  data-processing inequality} $I(X:Y) \ge I(X:Z)$. In fact,
\begin{equation}\label{eq:dataproc}
  I(X:Z)=I(X:Y)-I(X:Y|Z),
\end{equation}
and $I(X:Y|Z) \ge 0$.

An {\em ensemble of quantum states} $R = \{p_i,\rho_i\}_{i=1}^M$ is
represented by a set of $M$ density matrices $\rho_i$ (positive semidefinite
unit-trace operators), each with a prior probability $p_i$. For ensembles of
pure states we replace the density matrices with the normalized states, and we
write $V=\{p_i,\ket{\psi_i}\}_{i=1}^M$. A quantum measurement is described by
a {\em positive operator-valued measurement} (POVM) $\Pi = \{\Pi_j\}_{j=1}^N$,
defined as a set of $N$ positive semidefinite operators $\Pi_j$ that sum to
identity, namely $\sum_{j=1}^N\Pi_j=\mathbb{1}$. If we consider an ensemble
$R=\{p_i,\rho_i\}$ and a POVM $\Pi=\{\Pi_j\}$, the conditional probability
$p_{j|i}$ of outcome $j$ given the state $\rho_i$ is given by the {\em
  Born rule}, i. e. $p_{j|i} = \Tr[\rho_i\Pi_j]$. In the case of a POVM $\Pi$
performed over an ensemble $R$, the mutual information is a measure of how
much the outcomes of the POVM $\Pi$ are correlated with the states $\rho_i$,
in fact
\begin{equation}\label{eq:qminfo}
  I(R,\Pi) := \sum_{i,j} p_i \Tr[\rho_i\Pi_j] \log
  \frac{\Tr[\rho_i\Pi_j]}{\sum_k p_k\Tr[\rho_k\Pi_j] }.
\end{equation}
Now we can introduce the informational power of a POVM, the quantity that we
will analyze in the rest of this work.
\begin{dfn}
  The informational power $W(\Pi)$ of a POVM $\Pi$ is the maximum over all
  possible ensembles of states $R$ of the mutual information between $\Pi$ and
  $R$
  \begin{equation}\label{eq:infopower}
    W(\Pi) = \max_{R} I(R,\Pi).
  \end{equation}
  We call any ensemble that maximizes the mutual information a
  maximally informative ensemble for $\Pi$.
\end{dfn}

\subsection{Informational power as a classical capacity}\label{sec:additivity}

Given the tensor product $\otimes_{n=1}^N \Pi^n =\{\otimes_{n=1}^N
\Pi_{j_n}^n\}$ describing the parallel use of $N$ POVMs, by using entangled
input states one may ask if the informational power is superadditive. We
recall that the analogous quantity in the problem of optimization of POVMs,
namely the accessible information, is additive \cite{DW04}.

According to \cite{Hol98} (see also \cite{Kin01, Sho02}) we provide the
following definitions.
\begin{dfn}\label{def:capacity}
  Given a channel $\Phi$ from an Hilbert space $\mathcal{H}$ to an Hilbert
  space $\mathcal{K}$, the {\em single-use channel capacity} is given by
  \begin{equation}\label{eq:capacity}
    C_1(\Phi) := \sup_R \sup_\Lambda I(\Phi(R),\Lambda),
  \end{equation}
  where the suprema are taken over all ensembles $R$ in
  $\mathcal{H}$ and over all POVMs $\Lambda$ on $\mathcal{K}$.
\end{dfn}
\begin{dfn}\label{def:qc_channel}
  A {\em quantum-classical channel} (q-c channel) $\Phi_\Pi$ is defined as
  \begin{equation}
    \Phi_\Pi(\rho) := \sum_j \Tr[\rho \Pi_j] \ket{j}\bra{j}.
  \end{equation}
  where $\Pi=\{\Pi_j\}$ is a POVM and $\ket{j}$ is an orthonormal basis.
\end{dfn}
A q-c channel $\Phi_\Pi$ is a decision rule that maps quantum states into
classical states via a measurement $\Pi$.

\begin{thm}\label{lmm:capacity}
  The informational power of a POVM $\Pi=\{\Pi_j\}$ is equal to the single-use
  capacity
  $C_1(\Phi_\Pi)$ of the q-c channel $\Phi_\Pi$, i. e.
  \begin{equation}
    C_1(\Phi_\Pi) = W(\Pi).
  \end{equation}
\end{thm}
\begin{proof}
  Consider an ensemble $R=\{p_i,\rho_i\}$ and a POVM $\Lambda=\{\Lambda_k\}$.
  Introduce the random variables $X$, $Y$, and $Z$. Take $X$ with prior
  probability $p_i$. Take $Y$ such that the conditional probability of outcome
  $j$ of $Y$ given outcome $i$ of $X$ is $p_{j|i}=\Tr[\Pi_j\rho_i]$. Take $Z$
  such that the conditional probability of outcome $k$ of $Z$ given outcome
  $j$ of $Y$ is $q_{k|j}=\bra{j}\Lambda_k\ket{j}$. Clearly, the joint
  probability of outcome $i$ and $k$ of $X$ and $Z$ respectively is given by
  $p_i\Tr[\Lambda_k\Phi_\Pi(\rho_i)]$, so $I(X:Z) = I(\Phi_\Pi(R),\Lambda)$,
  whereas $I(X:Y) = I(R,\Pi)$.

  Notice that $X \to Y \to Z$ is a Markov chain, so Eq. (\ref{eq:dataproc})
  holds. By choosing $\Lambda_k=\ket{k}\bra{k}$, one has
  $q_{k|j}=\delta_{j,k}$, so $H(Y|Z)=0$, and $I(X:Y|Z) = H(Y|Z) - H(Y|X,Z) =
  0$ for any $\{p_i\}$. Thus,
  \begin{equation}
    \sup_\Lambda I(\Phi_\Pi(R),\Lambda) = I(\Phi_\Pi(R), \{\ket{k}\bra{k}\}).
  \end{equation}
  Since $p_i\bra{k}\Phi_\Pi(\rho_i)\ket{k}=p_i\Tr[\rho_i\Pi_k]$, we have
  \begin{equation}
    C_1(\Phi_\Pi) = \sup_R I(\Phi_\Pi(R),\{\ket{k}\bra{k}\}) = \sup_R I(R,\Pi) = W(\Pi).
  \end{equation}
\end{proof}

\begin{thm}\label{thm:additivity}
  The informational power $W(\Pi)$ is an additive quantity, i.e.
  \begin{equation}
    W(\otimes_{n=1}^N \Pi^n) = \sum_{n=1}^N W(\Pi^n).
  \end{equation}
\end{thm}
\begin{proof}
  Since the tensor product of q-c channels is a q-c channel,
  i. e. $\otimes_{n=1}^N\Phi_{\Pi^n} = \Phi_{\otimes_{n=1}^N\Pi^n}$, the
  statement follows immediately from Prop. \ref{lmm:capacity} and from the
  additivity property of the capacity for q-c channels \cite{Hol98,Sho02}.
\end{proof}

\subsection{Duality between informational power and accessible information}\label{sec:duality}

According to \cite{Hol73}, we provide the following definition.
\begin{dfn}
  The accessible information $A(R)$ of an ensemble $R=\{p_i,\rho_i\}$ is the
  maximum over all possible POVMs $\Pi$ of the mutual information between $R$
  and $\Pi$, namely
  \begin{equation}
    A(R) = \max_{\Pi} I(R,\Pi).
  \end{equation}
  We call any POVM that maximizes the mutual information a maximally
  informative POVM for $R$.
\end{dfn}

The accessible information of the ensemble $R=\{p_i,\rho_i\}$ is upper bounded
by the {\em Holevo quantity} \cite{Hol73},
\begin{equation}\label{eq:holevo}
  A(R) \le \chi(R) := S(\rho_R) - \sum_ip_iS(\rho_i),
\end{equation}
where $S(\rho) := -\Tr[\rho\log\rho]$ is the {\em von Neumann entropy} and
$\rho_R=\sum_i p_i \rho_i$. On the other hand, one has the following lower
bound \cite{JRW94}
\begin{equation}
  A(R) \ge Q(\rho_R) - \sum_ip_iQ(\rho_i),
\end{equation}
where $Q(\rho) := -\sum_k \left( \prod_{l \neq k} \frac{\lambda_k}{\lambda_k -
  \lambda_l} \right) \lambda_k \log \lambda_k, $ is the {\em subentropy} of a
quantum state, $\{\lambda_k\}$ being the set of eigenvalues of $\rho$.

Since invertible density matrices are a dense subset, in the following we
assume $\rho$ invertible. Given the ensemble $S=\{q_i,\sigma_i\}$, we call
$\sigma_S=\sum_iq_i\sigma_i$.
\begin{dfn}\label{def:dualens}
  Given an ensemble $S=\{q_i,\sigma_i\}$, we define the POVM $\Pi(S)$ as
  \begin{equation}\label{eq:dualens}
    \Pi(S):= \left\{q_i\sigma_S^{-1/2}\sigma_i\sigma_S^{-1/2}\right\}.
  \end{equation}
\end{dfn}
\begin{dfn}\label{def:dualpovm}
  Given a POVM $\Lambda=\{\Lambda_j\}$ and a density matrix $\sigma$, we define the
  ensemble $R(\Lambda,\sigma)$ as
  \begin{equation}\label{eq:dualpovm}
    R(\Lambda,\sigma) :=
    \left\{\Tr[\sigma\Lambda_j],
    \frac{\sigma^{1/2}\Lambda_j\sigma^{1/2}}{\Tr[\sigma\Lambda_j]}\right\}.
  \end{equation}
\end{dfn}
Definition \ref{def:dualens} corresponds to the so called ``pretty good''
measurement \cite{Bel75,HW94}. The ensemble-measurement duality given by the
definitions above has been exploited in \cite{Hal97} to obtain
measurement-dependent lower and upper bounds on $A(R(\Lambda ,\sigma))$. The
accessible information of the ensemble $R(\Lambda ,\sigma)$ has been studied
also in \cite{BH09}, in the context of quantifying the information-disturbance
tradeoff of quantum measurements.

In the following we show that there exists a duality between the informational
power and the accessible information that allows us to recast many results
from the latter context to the former one. Notice that $R(\Pi(S),\sigma_S) =
S$ and analogously $\Pi(R(\Lambda,\sigma)) = \Lambda$. Moreover, for any
ensemble $S$ and POVM $\Lambda$ one has
\begin{equation}\label{eq:dualmi}
  I(S,\Lambda) = I(R(\Lambda,\sigma_S),\Pi(S)).
\end{equation}

\begin{thm}\label{thm:duality}
  The informational power of a POVM $\Lambda=\{\Lambda_j\}$ is given by
  \begin{equation}
    W(\Lambda) = \max_\sigma A(R(\Lambda,\sigma)).
  \end{equation}
  The ensemble $S^*=\{q_i^*,\sigma_i^*\}$ is maximally informative for the
  POVM $\Lambda$ if and only if $\sigma_{S^*}=\arg\max_\sigma
  A(R(\Lambda,\sigma))$ and the POVM $\Pi(S^*)$ is maximally informative for
  the ensemble $R(\Lambda,\sigma_{S^*})$.
\end{thm}
\begin{proof}
  From the definitions of informational power and accessible information, and
  from Eq. \ref{eq:dualmi}, one has
  \begin{equation}
    \begin{split}
      W(\Lambda) & =\max_\sigma\max_{S|\sigma_S=\sigma}I(S,\Lambda)\\ & =
      \max_\sigma\max_{\Pi(S)|\sigma_S=\sigma}I(R(\Lambda,\sigma_S),\Pi(S))\\
      & = \max_\sigma\max_\Pi I(R(\Lambda,\sigma),\Pi)\\
      & = \max_\sigma A(R(\Lambda,\sigma)).
    \end{split}
  \end{equation}
\end{proof}

Proposition \ref{thm:duality} makes clear the duality between the
informational power and the accessible information. A diagrammatic
representation of this duality is given by
\begin{equation*}
  \begin{CD}
    \Lambda @>\sigma_{S^*}>> R(\Lambda,\sigma_{S^*})\\ @VVV @VVV\\ S^*
    @<\sigma_{S^*}<< \Pi(S^*)
  \end{CD}
\end{equation*}
where $S^*=\arg \max_S I(S,\Lambda)$ and $\Pi(S^*)=\arg \max_\Pi
I(R(\Lambda,\sigma_{S^*}),\Pi)$. Horizontal arrows correspond to the duality
operation of Definitions \ref{def:dualens} and \ref{def:dualpovm}. Moving in
the sense of the arrow corresponds to apply Eq. (\ref{eq:dualpovm}), thus
requiring $\sigma_{S^*}$. Moving in the opposite sense corresponds to apply
Eq. (\ref{eq:dualens}). The vertical arrow from $\Lambda$ to $S^*$ indicates
that $S^*$ is maximally informative for the POVM $\Lambda$, whereas the
vertical arrow from $R(\Lambda,\sigma_{S^*})$ to $\Pi(S^*)$ indicates that
$\Pi(S^*)$ is maximally informative for the ensemble
$R(\Lambda,\sigma_{S^*})$.

From Prop. \ref{thm:duality} we can obtain a property of maximally informative
ensembles using Davies' theorem \cite{Dav78}.
\begin{thm}\label{thm:davies}
  Given a $D$-dimensional POVM $\Lambda=\{\Lambda_j\}$, there exists a
  maximally informative ensemble $S^*=\{q_i^*,\sigma_i^*\}_{i=1}^M$, with all
  $\sigma_i^*$ pure and $D\le M\le D^2$.
\end{thm}
\begin{proof}
  By Prop. \ref{thm:duality}, $S^*$ is maximally informative for $\Lambda$ if
  and only if $\sigma_{S^*}=\arg\max_\sigma A(R(\Lambda,\sigma))$ and
  $\Pi(S^*)$ is maximally informative for $R(\Lambda, \sigma_{S^*})$. By
  Davies' theorem \cite{Dav78}, there exists a maximally informative POVM
  $\Pi(S^*)$ with $M$ rank-one elements and $D\le M \le D^2$, so the statement
  follows.
\end{proof}

For some classes of POVMs it is possible to improve the bound on the number of
elements of a maximally informative ensemble as follows.
\begin{dfn}
  An ensemble $S=\{q_i,\sigma_i\}$ on an Hilbert space $\mathcal{H}$ is real
  if there exists a basis on $\mathcal{H}$ relative to which all $\sigma_i$
  have real matrix elements.
\end{dfn}
\begin{dfn}
  A POVM $\Lambda=\{\Lambda_j\}$ on an Hilbert space $\mathcal{H}$ is real if
  there exists a basis on $\mathcal{H}$ relative to which all $\Lambda_j$ have
  real matrix elements.
\end{dfn}
\begin{thm}\label{thm:sasaki}
  Given a $D$-dimensional real POVM $\Lambda=\{\Lambda_j\}$, there exists a
  maximally informative real ensemble $S^*=\{q_i^*,\sigma_i^*\}_{i=1}^M$, with
  all $\sigma_i^*$ pure and $D\le M\le D(D+1)/2$.
\end{thm}
\begin{proof}
  By Prop. \ref{thm:duality}, $S^*$ is maximally informative for $\Lambda$ if
  and only if $\sigma_{S^*}=\arg\max_\sigma A(R(\Lambda,\sigma))$ and
  $\Pi(S^*)$ is maximally informative for $R(\Lambda, \sigma_{S^*})$. By Lemma
  5 of \cite{SBJOH99}, there exists a maximally informative POVM $\Pi(S^*)$
  with $M$ rank-one elements and $D\le M \le D(D+1)/2$, so the statement
  follows.
\end{proof}

\section{Evaluation of the informational power}\label{sec:evaluation}

Given a POVM, it is in general an hard task to provide an explicit form for
the maximally informative ensemble, due to the non-linearity of the mutual
information as a figure of merit. In the following, we prove some necessary
conditions for attaining informational power, and we make use of these results
to provide an iterative algorithm converging to the maximally informative
ensemble. In this section it is convenient to take the states of the ensemble
unnormalized, with the norm giving the prior probability of each
state. Therefore we will also use the notation for the ensemble $V :=
\{\ket{\psi_i}\}$, with prior probability $p_i=||\psi_i||^2$.

\subsection{Necessary conditions to attain informational power}\label{sec:conditions}

When one optimizes the informational power, considering only ensembles of pure
states is not restrictive, as shown in Prop. \ref{thm:davies}. We provide here
a short alternative proof of this fact, which is independent of Davies'
theorem \cite{Dav78}.

\begin{thm}\label{thm:pures}
  For any given POVM $\Pi=\{\Pi_j\}$, there exists a maximally informative
  ensemble made of pure states.
\end{thm}
\begin{proof}
  Consider an ensemble $R=\{p_i,\rho_i\}$. Each of the states can be
  decomposed on the basis of its orthogonal eigenvectors as $\rho_i = \sum_k
  \ket{\psi_{ik}}\bra{\psi_{ik}}$, with $\sum_k ||\psi_{ik}||^2=1$, $\forall
  i$. Denote by $V=\{\ket{\psi_{ik}}\}$ the ensemble of such pure states.

  For three random variables $X$, $Y$, and $Z$, we have
  \begin{equation}
    \begin{split}
      I(X:Z) & = H(Z) - H(Z|X)\\
      & \le H(Z) - H(Z|X,Y) = I(X,Y:Z),
    \end{split}
  \end{equation}
  since conditioning reduces entropy.  We take $X$ distributed according to
  $p_i$. If we set the joint probability $p_{i,j}$ of outcome $i$ of $X$ and
  $j$ of $Z$ to be $p_{i,j}=p_i\Tr[\Pi_j\rho_i]$, we have $I(X:Z)=I(R,\Pi)$.
  If we set the joint probability $p_{i,k,j}$ of outcome $i$, $k$, $j$ of $X$,
  $Y$ and $Z$, respectively, to be
  $p_{i,k,j}=p_i\bra{\psi_{i,k}}\Pi_j\ket{\psi_{i,k}}$, we have
  $I(X,Y:Z)=I(V,\Pi)$, and hence $I(V,\Pi) \ge I(R,\Pi)$. Clearly, the maximum
  of $I(R,\Pi)$ over $R$ can be searched only among ensembles of pure states.
\end{proof}

Now we turn to the problem of finding necessary conditions for an ensemble of
pure states to be maximally informative for a given POVM $\Pi=\{\Pi_j\}$. For
any ensemble $V=\{\ket{\psi_i}\}$, by defining
\begin{equation}\label{eq:piprime}
  \Pi'_i := \sum_{j=1}^N \log \frac{\bra{\psi_i} \Pi_j
    \ket{\psi_i}}{||\psi_i||^2\sum_{k=1}^M \bra{\psi_k} \Pi_j
    \ket{\psi_k}} \Pi_j,
\end{equation}
we notice that the mutual information $I(V,\Pi)$ can be written as
$I(V,\Pi)=\sum_i\bra{\psi_i}\Pi'_i\ket{\psi_i}$.

\begin{thm}\label{thm:condition1}
  Given a POVM $\Pi=\{\Pi_j\}$, a necessary condition for the ensemble
  $V=\{\ket{\psi_i}\}_{i=1}^M$ to be maximally informative is that
  \begin{equation}\label{eq:condition1}
    \Pi'_i \ket{\psi_i}= I(V,\Pi) \ket{\psi_i} \qquad \forall i=1,\dots M,
  \end{equation}
  where $\Pi'_i$ is given in Eq. (\ref{eq:piprime}).
\end{thm}
\begin{proof}
  Upon introducing a Lagrange multiplier $\lambda$ in order to constrain the
  normalization of the input ensemble, let us consider the expression
  \begin{equation}\label{eq:merit}
    C = \sum_{i=1}^M \bra{\psi_i} \Pi'_i \ket{\psi_i} +
    \lambda\left(\sum_{i=1}^M ||\psi_i||^2 - 1\right).
  \end{equation}
  By equating to zero the derivative of Eq. (\ref{eq:merit})
  with respect to each $\bra{\psi_i}$, we obtain $M$
  extremal equations which are necessary conditions for a maximally
  informative ensemble, namely
  \begin{equation}
    \frac{\partial C}{\partial \bra{\psi_i}} = [\Pi'_i +
    (\lambda-1)\mathbb{1}] \ket{\psi_i} = 0, \qquad \forall
    i = 1,\dots M.
  \end{equation}
  Upon redefining $\mu = 1 - \lambda$, we can rewrite the extremal equations
  as $\Pi'_i \ket{\psi_i} = \mu \ket{\psi_i}$. By multiplying both sides on
  the left by $\ket{\psi_i}$ and summing over $i$, we notice that
  $\mu=I(V,\Pi)$.
\end{proof}

\begin{cor}\label{thm:condition2}
  Given a POVM $\Pi=\{\Pi_j\}$, a necessary condition for
  $V=\{\ket{\psi_i}\}_{i=1}^M$ to be maximally informative is that
  \begin{equation}\label{eq:minfocond}
    I(V,\Pi) = \sqrt{\sum_{i=1}^M \bra{\psi_i} {\Pi'_i}^2 \ket{\psi_i}}.
  \end{equation}
\end{cor}
\begin{proof}
  The result follows immediately by multiplying Eq. (\ref{eq:condition1}) on
  the left by its Hermitian adjoint, and summing over $i$.
\end{proof}

\subsection{An iterative algorithm to maximize informational power}\label{sec:algorithm}

In the following we provide a steepest-ascent iterative algorithm which is
effective in finding a maximally informative ensemble for a given POVM. A
similar algorithm for the evaluation of the accessible information for a given
ensemble can be found in \cite{REK05}.
\begin{algo}\label{thm:iteration}
  The following steepest-ascent algorithm converges to a maximum of the
  informational power. For arbitrary ensemble
  $V^0=\{\ket{\psi_i^0}\}_{i=1}^M$, evaluate $V^n=\{\ket{\psi_i^n}\}_{i=1}^M$
  at any order $n$ by the following steps:
  \begin{enumerate}
  \item
    Given $V^n=\{\ket{\psi_i^n}\}_{i=1}^M$, evaluate
    ${\Pi'}^n=\{{\Pi'_i}^n\}_{i=1}^M$ according to
    \begin{equation}\label{eq:step1}
      {\Pi'_i}^n = \sum_{j=1}^N \log \frac{\bra{\psi_i^n} \Pi_j
        \ket{\psi_i^n}}{\sum_{k=1}^M \bra{\psi_k^n} \Pi_j \ket{\psi_k^n}}
      \Pi_j - \log ||\psi_i^n||^2 \mathbb{1}.
    \end{equation}
  \item
    Pick up a small enough positive $\alpha$ and evaluate
    \begin{equation}\label{eq:step2}
      \ket{\hat{\psi}_i^{n+1}} = [(1-\alpha)\mathbb{1} +
      \alpha {\Pi'_i}^n] \ket{\psi_i^n}.
    \end{equation}
  \item
    Obtain $V^{n+1}$ as
    \begin{equation}\label{eq:step3}
      \ket{\psi_i^{n+1}} = \frac{\ket{\hat{\psi}_i^{n+1}}} {\sqrt{\sum_{i=1}^M
          ||\hat{\psi}_i^{n+1}||^2}}.
    \end{equation}
  \end{enumerate}
\end{algo}
\begin{proof}
  Consider the POVM $\Pi=\{\Pi_j\}$ and an ensemble
  $V^n=\{\ket{\psi_i^n}\}_{i=1}^M$, so Eq. (\ref{eq:step1}) is just the
  definition given in (\ref{eq:piprime}).

  The algorithm we are considering is a steepest-ascent algorithm. We move the
  ensemble in the direction of the gradient of the mutual information, namely
  \begin{equation}
    \begin{split}
      \nabla I(V,\Pi) & = \left(\frac{\partial I}{\partial\bra{\psi_1}},\dots
      \frac{\partial I}{\partial\bra{\psi_M}} \right)\\
      & = \left((\Pi'_1-\mathbb{1})\ket{\psi_1},\dots
      (\Pi'_M-\mathbb{1})\ket{\psi_M}\right),
    \end{split}
  \end{equation}
  which ensures that we follow the greatest increase of the mutual
  information. 
  So, if we set the iteration to be
  \begin{equation}\label{eq:proto}
    \left(\ket{\hat\psi_1^{n+1}},\dots\ket{\hat\psi_M^{n+1}}\right) =
    (1-\alpha)\left(\ket{\psi_1^{n}},\dots\ket{\psi_M^{n}}\right) + \alpha
    \nabla I(\Pi,V^n),
  \end{equation}
  we obtain Eq. (\ref{eq:step2}).

  Then, Eq. (\ref{eq:step3}) is just the normalization of the updated ensemble
  $V^{n+1}$ in order to satisfy $\sum_{i=1}^M||\psi_i||^2=1$. By construction,
  one has $I(V^{n+1},\Pi) \ge I(V^n,\Pi)$.
\end{proof}

As for all steepest-ascent algorithm, there is no protection against the
possibility of convergence toward a local, rather than a global, maximum,
whence one should run the algorithm for different initial ensembles in order
to discriminate between local and global maxima.

Any ensemble can be used as a starting point, except for a subset
corresponding to minima of the mutual information (for example, all the
ensembles composed by a single quantum state). These minima are unstable fix
points of the iteration, so even small perturbations let the iteration
converge to some maximum. Due to Propositions \ref{thm:davies} and
\ref{thm:sasaki}, it is sufficient to consider ensembles with $D^2$ states for
a $D$-dimensional POVM, and with $D(D+1)/2$ states for a real POVM.

The parameter $\alpha$ controls the length of each iterative step, so for
$\alpha$ too large, an overshooting can occur. This can be kept under control
by evaluating the mutual information $I(V,\Pi)$ at the end of each step: if
$I(V,\Pi)$ decreases instead of increasing, we are warned that we have taken
$\alpha$ too large. An efficient evaluation of $I(V,\Pi)$ can be performed
through Corollary \ref{thm:condition2}.

\section{Classification of quantum measurements}\label{sec:examples}

The informational power introduces a complete ordering between POVMs. In the
following, we classify some POVMs according to their informational power. We
will consider POVMs with commuting elements (Sect. \ref{sec:commuting}),
real-symmetric POVMs (Sect. \ref{sec:realsymm}), mirror-symmetric POVMs
(Sect. \ref{sec:mirrorsymm}), and the $2$-dimensional symmetric
informationally complete POVM (Sect. \ref{sec:tetrahedron}),

\subsection{POVMs with commuting elements}\label{sec:commuting}

\begin{thm}\label{thm:commpovm}
  Given a $D$-dimensional POVM $\Pi=\{\Pi_j\}_{j=1}^N$ with commuting
  elements, there exists a maximally informative ensemble
  $V=\{p_i^*,\ket{i}\}_{i=1}^M$ of $M \le D$ states, where $\ket{i}$ denotes
  the common orthonormal eigenvectors of $\Pi$, and the prior probabilities
  $p_i^*$ maximize the mutual information
  \begin{equation}\label{eq:prior_minfo}
    W(\Pi) = \max_{p_i} \sum_{i,j} p_i \bra{i}\Pi_j\ket{i}
    \log\frac{ \bra{i}\Pi_j\ket{i} }{\sum_k p_k \bra{k}\Pi_j\ket{k}}.
  \end{equation}
\end{thm}
\begin{proof}
  For any ensemble $R=\{p_i,\rho_i\}$, consider the diagonal ensemble
  $S=\{p_i,\sigma_i\}$, where $\sigma_i=\sum_k
  \bra{k}\rho_i\ket{k}\ket{k}\bra{k}$ with $\ket{k}$ denoting the common
  eigenvectors of $\Pi$.  Clearly, $\Tr[\Pi_j \sigma_i] = \Tr[\Pi_j \rho_i]$,
  whence $I(R,\Pi)=I(S,\Pi)$.  As in Prop. {\ref{thm:pures}}, it is sufficient
  to look for the maximum over the prior probabilities $p_i$, with fixed
  states $\ket{i}$.  Hence Eq. (\ref{eq:prior_minfo}) follows.
\end{proof}
We notice that $M \le D$ since some of the prior $p_i$ obtained by optimizing
Eq. (\ref{eq:prior_minfo}) can be zero. Equation (\ref{eq:prior_minfo}) is a
concave function of the prior probabilities, and a numerical algorithm for
performing the optimization is provided in \cite{Bla72}.

As an application, we consider the POVM
$\Pi^{(\eta)}=\{\Pi_j^{(\eta)}\}_{j=1}^D$ describing the projective
measurement over an orthonormal basis $\{\ket{j}\}$ in dimension $D$ affected
by isotropic noise, i. e.
\begin{equation}\label{eq:onbdep}
  \Pi_j^{(\eta)} = \eta\ket{j}\bra{j} + (1-\eta) \frac{\mathbb{1}}{D}, \qquad
  j=1,\dots D.
\end{equation}
When $\eta=1$, a maximally informative ensemble is clearly
$V=\{p_i,\ket{i}\}$, with $p_i=1/D$. For $\eta<1$, by
Prop. \ref{thm:commpovm}, the ensemble $V$ is maximally informative for
$\{p_i\}$ maximizing Eq. (\ref{eq:prior_minfo}). By Born rule, the conditional
probability $p_{j|i}$ of outcome $j$ given the state $\ket{i}$ is
$p_{j|i}=\eta\delta_{i,j}+\frac{1-\eta}{D}$. Consider two random variables $X$
and $Y$ with joint probability $p_{i,j}=p_ip_{j|i}$ and marginal probabilities
$p_i$ and $q_j=\sum_ip_ip_{j|i}$, respectively. Clearly, $I(X:Y) =
I(V,\Pi^{(\eta)})$. If $p_i=\frac1D$, then $q_j=\frac1D$, and the Shannon
entropy $H(Y)$ of $Y$ is obviously maximized, i. e. $H(Y) = \log D$.
Moreover, the conditional Shannon entropy $H(Y|X)$ is independent of $p_i$,
and in fact one has
\begin{equation}
  \begin{split}
  H(Y|X) = & -
  \left(\eta+\frac{1-\eta}{D}\right)\log\left(\eta+\frac{1-\eta}{D}\right)\\ &
  - (D-1)\frac{1-\eta}{D}\log\frac{1-\eta}{D}.
  \end{split}
\end{equation}
Since $I(X:Y)=H(Y)-H(Y|X)$, the maximum of the mutual information is attained
for $p_i=\frac1D$, and the informational power is
$W(\Pi^{(\eta)})=\log(D)-H(Y|X)$. As expected, the informational power is an
increasing function of $\eta$, and is plotted in Fig. \ref{fig:onbdep}, for
different values of $D$.
\begin{figure}[htb]
  \includegraphics{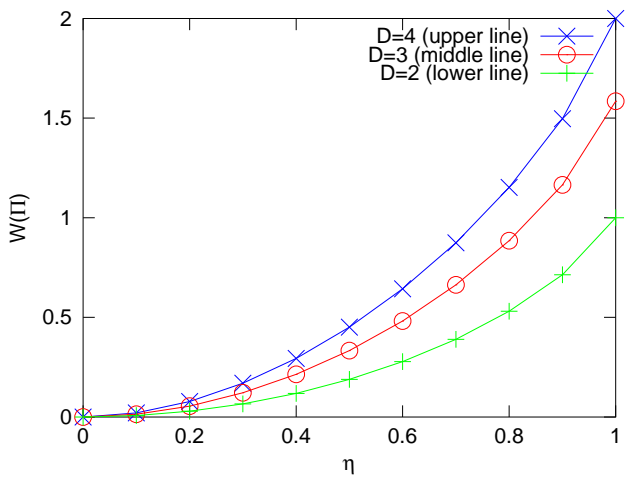}
  \caption{(Color online) Informational power $W(\Pi)$ of the $D$-dimensional
    POVM $\Pi^{(\eta)}$ projecting on the orthonormal basis affected by
    isotropic noise parameterized by $\eta$ [see Eq. (\ref{eq:onbdep})], as a
    function of $\eta$, for dimension $D=2,3,4$.}
  \label{fig:onbdep}
\end{figure}

This result can be useful to prove that the protocols proposed in \cite{DDS10}
for the purification of noisy quantum measurements are indeed optimal. The aim
of purification of noisy quantum measurements is to recast many uses of a
noisy POVM to a single use of an ideal POVM. More precisely, given an ensemble
$R$ and $N$ uses of a noisy POVM $\Pi$, one can ask what channel $\Phi$
maximizes the mutual information $I(\Phi(R),\Pi^{\otimes N})$. For example,
suppose that we have the ensemble $V=\{1/D,\ket{i}\}_{i=1}^D$ and $N$ uses of
the $D$-dimensional noisy POVM $\Pi^{(\eta)}$ as in
Eq. (\ref{eq:onbdep}). Since we have shown that the maximally informative
ensemble for $\Pi^{(\eta)}$ is $V$, by Prop. \ref{thm:additivity}, the channel
$\Phi$ that maximizes $I(\Phi(V),\Pi^{(\eta)\otimes N})$ is the orthogonal
cloning,
i. e. $\Phi(\rho)=\sum_{i=1}^{D}\bra{i}\rho\ket{i}(\ket{i}\bra{i})^{\otimes
  N}$.

\subsection{Real-symmetric POVMs}\label{sec:realsymm}

In the following we parameterize any pure state as
$\ket{\psi}=\pure{\cos\theta}{\sin\theta}$, in the basis of the eigenvectors
$\ket{0}$ and $\ket{1}$ of the Pauli matrix $\sigma_z$. We denote with $Z_N$
the group of rotations of $\pi/N$ around the $y$-axis, generated by $U =
\exp\left(-i\frac{\pi}{N}\sigma_y\right)$.
\begin{dfn}
  A $2$-dimensional real ensemble $V=\{p_i,\ket{\psi_i}\}_{i=0}^{M-1}$, with
  $\ket{\psi_i}=U^i\ket{\psi_0}$ for any fixed $\ket{\psi_0}$, is called
  real $Z_M$-symmetric.
\end{dfn}
\begin{dfn}
  A $2$-dimensional real POVM $\Pi=\{\Pi_j\}_{j=0}^{N-1}$, with $\Pi_j=\frac2N
  \ket{\pi_j}\bra{\pi_j}$ and $\ket{\pi_j}=U^j\ket{\pi_0}$ for any fixed
  $\ket{\pi_0}$, is called real $Z_N$-symmetric.
\end{dfn}
Without loss of generality, we will take $\ket{\pi_0}=\ket{0}$.

\begin{thm}\label{thm:iprs}
  For any real $Z_N$-symmetric POVM $\Pi=\{\frac2N
  \ket{\pi_j}\bra{\pi_j}\}_{j=0}^{N-1}$, the ensemble $V=\{p_i,
  \ket{\psi_i}\}_{i=0}^{M-1}$, with
  $\ket{\psi_i}=\pure{\sin{\theta_i}}{\cos{\theta_i}}$, is maximally
  informative if $M$, $\{\theta_i\}$ and $\{p_i\}$ are taken as either
  \begin{itemize}
    \item (real $Z_N$-symmetric) $M=N$, $\theta_i=\frac{\pi i}{N}$ and
      $p_i=\frac1N$,
    \item (real $Y$-shaped) $M=3$, $\theta_0=0$,
      $\theta_1=\frac{\pi n}{N}$, $\theta_2=-\frac{\pi n}{N}$, and
      $p_0=1-2p_1$, $p_1=p_2=\frac{1}{4\sin^2{\frac{\pi n}N}}$, $\forall n$
      such that $0\le p_0\le1$.
  \end{itemize}
  The informational power of $\Pi$ is given by
  \begin{equation}
    W(\Pi) = \sum_{j=0}^{N-1} \left[\frac2N\sin^2\left(\frac{\pi
        j}{N}\right)\right] \log\left[\frac2N\sin^2\left(\frac{\pi
        j}{N}\right)\right] + \log N.
  \end{equation}
\end{thm}
\begin{proof}
  The conditional probability $p_{j|i}$ of outcome $j$ given the state
  $\ket{\psi_i}$ is $p_{j|i}=\frac2N\sin^2(\theta_i-\frac{\pi j}N)$, and the
  probability $q_j$ of outcome $j$ is $q_j=\sum_{i=0}^{M-1}p_ip_{j|i}$.

  Consider the random variables $X$ and $Y$, with $X$ distributed according to
  $p_i$, and $Y$ such that the conditional probability of outcome $j$ of $Y$
  given outcome $i$ of $X$ is $p_{j|i}$. Clearly $I(X:Y)=I(V,\Pi)$.

  By setting $f(\theta_i)=\sum_{j=0}^{N-1}p_{j|i}\log p_{j|i}$, we have for
  the joint entropy $H(Y|X)=-\sum_{i=0}^{M-1}p_if(\theta_i)$. As shown in
  Lemma 3 of \cite{SBJOH99}, $f(\theta)$ attains its global maximum for
  $\theta=\frac{\pi k}{N}$, $k\in\mathbb{N}$. Thus by choosing $\{\theta_i\}$
  multiples of $\frac{\pi}{N}$, $H(Y|X)$ attains its minimum $H(Y|X)=f(0)$,
  independent of $M$ and $\{p_i\}$.

  By taking the real $Z_N$-symmetric or the real $Y$-shaped parameterizations
  for $M$, $\{\theta_i\}$ and $\{p_i\}$, we have $q_j=\frac1N$, so the entropy
  $H(Y)$ attains its maximum, i. e. $H(Y)=\log N$. Since $I(X:Y)=H(Y)-H(Y|X)$,
  the Proposition remains proved.
\end{proof}

We notice that for a real $Z_N$-symmetric POVM $\Pi=\{\frac2N
\ket{\pi_j}\bra{\pi_j}\}$, any maximally informative ensemble
$V=\{p_i,\ket{\psi_i}\}$ given in Proposition \ref{thm:iprs} is such that
every state $\ket{\psi_i}$ is orthogonal to one of the
$\ket{\pi_j}$. Considering the real $Y$-shaped parameterization, we observe
that if $N$ is even one can chose $n=\frac{N}2$, obtaining
$V=\{1/2,\ket{i}\}$, with $i=0,1$. With this choice, the maximally informative
real $Y$-shaped ensemble is minimal. For some real $Z_N$-symmetric POVMs, the
maximally informative ensembles with minimal number of states are represented
in Fig. \ref{fig:rsensembles}.
\begin{figure}[htb]
  \includegraphics{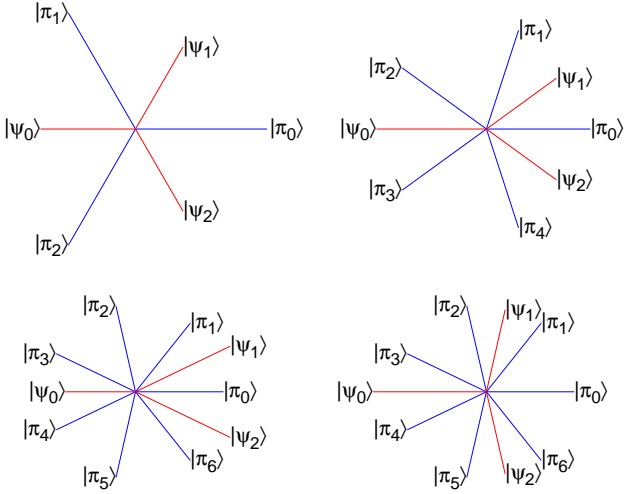}
  \caption{(Color online) Real $Z_N$-symmetric POVMs $\Pi=\{\frac2N
    \ket{\pi_j}\bra{\pi_j}\}_{j=0}^{N-1}$ (blue vectors labeled by
    $\ket{\pi_j}$) for $N=3$ (upper-left), $N=5$ (upper-right), and $N=7$
    (lower-left and lower-right). A maximally informative ensemble (red
    vectors labeled by $\ket{\psi_i}$) $V=\{p_i,\ket{\psi_i}\}_{i=0}^{M-1}$
    with minimal number of states is represented for each POVM. For $N=7$,
    there are two inequivalent maximally informative ensembles (lower-left and
    lower-right). The angle between the states
    $\pure{\cos\theta_0}{\sin\theta_0}$ and
    $\pure{\cos\theta_1}{\sin\theta_1}$ is $2(\theta_1-\theta_0)$, as in the
    Bloch sphere representation. The length of the vector corresponding to
    state $\ket{\psi_i}$ is proportional to $\sqrt{p_i}$.}
  \label{fig:rsensembles}
\end{figure}

The real $Z_3$-symmetric POVM $\Pi$ is usually called {\em trine}
measurement. The informational power of $\Pi$ is $W(\Pi)=\log3/2$ by
Prop. \ref{thm:iprs}. The maximally informative ensemble for $\Pi$
parameterized as in Prop. \ref{thm:iprs} is usually called {\em
  antitrine}. The analogous problem of maximization of the accessible
information for real-symmetric ensembles has been addressed by Holevo
\cite{Hol73} and by Sasaki et al. \cite{SBJOH99}.

\subsection{Mirror-symmetric POVMs}\label{sec:mirrorsymm}

In this subsection we apply the duality shown in Prop. \ref{thm:duality}
between the informational power and the accessible information to
mirror-symmetric POVMs.

\begin{dfn}
  We call mirror-symmetric ensemble any $2$-dimensional real ensemble
  $S=\{p_i,\ket{\psi_i}\}$ such that for any $\ket{\psi_i}$, there exists a
  $\ket{\psi_k}=\sigma_z\ket{\psi_i}$ and $p_i=p_k$.
\end{dfn}
\begin{dfn}
  We call mirror-symmetric POVM any $2$-dimensional real POVM
  $\Lambda=\{\Lambda_j\}$ with $\Lambda_j=n_j\ket{\lambda_j}\bra{\lambda_j}$
  such that for any $\ket{\lambda_j}$, there exists a
  $\ket{\lambda_l}=\sigma_z\ket{\lambda_j}$ and $n_j=n_l$.
\end{dfn}

The problem of accessible information for mirror-symmetric POVMs has been
addressed in \cite{Fre06}. From Definitions \ref{def:dualens} and
\ref{def:dualpovm}, it immediately follows that if the ensemble $S$ is
mirror-symmetric, the POVM $\Pi(S)$ is mirror-symmetric, and that if the POVM
$\Lambda$ is mirror-symmetric, the ensemble $R(\Lambda,\sigma)$ is
mirror-symmetric, for any density matrix $\sigma$.

\begin{thm}
  Given a mirror-symmetric POVM $\Lambda=\{\Lambda_j\}$, there exists a
  maximally informative ensemble $S=\{p_i,\ket{\psi_i}\}_{i=0}^{M-1}$ such
  that $S$ is mirror-symmetric and $M\le4$.
\end{thm}
\begin{proof}
  By Prop. \ref{thm:duality}, $S^*$ is maximally informative for $\Lambda$ if
  and only if $\sigma_{S^*}=\arg\max_\sigma A(R(\Lambda,\sigma))$ and
  $\Pi(S^*)$ is maximally informative for $R(\Lambda, \sigma_{S^*})$. By
  Prop. 2 in \cite{Fre06}, there exists a maximally informative
  mirror-symmetric four-element POVM $\Pi(S^*)$, so the statement follows.
\end{proof}

As an application we consider the mirror-symmetric POVM
$\Pi=\{n_j\ket{\pi_j}\bra{\pi_j}\}_{j=0}^2$, with
\begin{equation}\label{eq:y}
  \begin{array}{lll}
    \ket{\pi_0} = \pure{1}{0}, & 
    \ket{\pi_1} = \pure{\sin\theta}{\cos\theta}, &
    \ket{\pi_2} = \pure{\sin\theta}{-\cos\theta},
  \end{array}
\end{equation}
and $n_0=\frac{\cos2\theta}{\cos^2\theta}$ and
$n_1=n_2=\frac{1}{2\cos^2\theta}$.  Figure \ref{fig:Y} shows the informational
power $W(\Pi)$ as a function of $\theta$, as obtained by Algorithm
\ref{thm:iteration}.
\begin{figure}[htb]
  \includegraphics{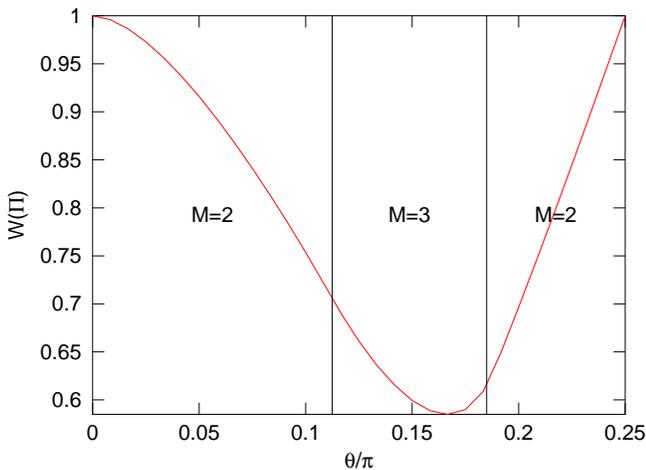}
  \caption{(Color online) Informational power of a mirror-symmetric POVM
    $\Pi=\{n_j\ket{\pi_j}\bra{\pi_j}\}_{j=0}^2$, with $\ket{\pi_j}$ as in
    Eq. (\ref{eq:y}), as a function of $\theta$. The minimum is attained for
    $\theta=\pi/6$, where $\Pi$ corresponds to the trine POVM and a maximally
    informative ensemble is the antitrine. The maxima at $\theta=0$ and
    $\theta=\pi/4$ correspond to the degenerate case of the POVM projecting on
    orthonormal basis. The label $M=2,3$ denotes the minimum number $M$ of
    states of a maximally informative ensemble in each of the three regions.}
  \label{fig:Y}
\end{figure}

\subsection{$2$-dimensional SIC POVM}\label{sec:tetrahedron}

According to \cite{RBSC04, AFF10}, we provide the following definition.
\begin{dfn}
  A $D$-dimensional POVM $\Pi=\{\Pi_j\}_{j=0}^{N-1}$ with $N=D^2$ elements
  $\Pi_j = \frac{1}{D} \ket{\pi_j}\bra{\pi_j}$ with invariant inner product
  $\Tr[\Pi_j\Pi_l] = \left(D^2(D+1)\right)^{-1}$, for any $i \ne j$, is called
  symmetric informationally complete (SIC) POVM.
\end{dfn}

For $D=2$ there exists only one SIC POVM $\Pi=\{\frac12
\ket{\pi_j}\bra{\pi_j}\}_{j=0}^3$ with
\begin{equation}\label{eq:tetra}
  \begin{array}{ll}
    \ket{\pi_0} = \pure{1}{0}, & \ket{\pi_1} =
    \pure{\frac{1}{\sqrt3}}{\sqrt{\frac23}},\\
    \ket{\pi_2} = \pure{\frac{1}{\sqrt3}}{e^{i\frac23\pi}\sqrt{\frac23}}, &
    \ket{\pi_3} = \pure{\frac{1}{\sqrt3}}{e^{i\frac43\pi}\sqrt{\frac23}}.
  \end{array}
\end{equation}
Since these states lie on the four vertex of a tetrahedron, this POVM is
usually called the {\em tetrahedron}.

\begin{thm}\label{thm:tetra}
  Given the $2$-dimensional SIC POVM
  $\Pi=\{\frac12\ket{\pi_j}\bra{\pi_j}\}_{j=0}^3$, the ensemble $V=\{\frac14,
  \ket{\psi_i}\}_{i=0}^3$ with
  \begin{equation}\label{eq:antitetra}
    \begin{array}{ll}
      \ket{\psi_0} = \pure{0}{1}, & \ket{\psi_1} =
      \pure{\sqrt{\frac23}}{-\frac{1}{\sqrt3}}\\ \ket{\psi_2} =
      \pure{\sqrt{\frac23}}{e^{i\frac13\pi}\frac{1}{\sqrt3}}, & \ket{\psi_3} =
      \pure{\sqrt{\frac{2}{3}}} {e^{i\frac53\pi}\frac{1}{\sqrt{3}}}.
    \end{array}
  \end{equation}
  is maximally informative, and the informational power is $W(\Pi) =
  \log\frac{4}{3}$.
\end{thm}
\begin{proof}
  Consider an ensemble $V=\{p_i,\ket{\psi_i}\}$ parameterized as
  $\ket{\psi_i}=\pure{\sin\theta_i}{e^{i\phi_i}\cos\theta_i}$. Call
  $p_{j|i}=|\braket{\psi_i}{\pi_j}|^2$ the conditional probability of outcome
  $j$ given the state $\ket{\psi_i}$, and $q_j=\sum_{i=0}^3p_ip_{j|i}$ the
  probability of outcome $j$.

  Consider the random variables $X$ and $Y$, with $X$ distributed according to
  $p_i$, and $Y$ such that the conditional probability of outcome $j$ of $Y$
  given outcome $i$ of $X$ is $p_{j|i}$. Clearly, $I(X:Y)=I(V,\Pi)$.

  By setting $f(\theta_i,\phi_i)=\sum_{j=0}^{N-1}p_{j|i}\log p_{j|i}$, we have
  for the joint entropy $H(Y|X)=-\sum_{i=0}^{M-1}p_if(\theta_i,\phi_i)$. As it
  is easy to show, $f(\theta,\phi)$ attains its global maximum $\log3$ at
  $\theta=0$ for any $\phi$, and at $\theta= \arccos(\frac1{\sqrt{3}})$ for
  $\phi=\frac\pi3$, $\phi=\pi$, and $\phi=\frac{5\pi}3$. Thus making one of
  these choices for $\{\theta_i,\phi_i\}$, $H(Y|X)$ attains its minimum
  $H(Y|X)=\log3$.

  Moreover, by setting $M=4$ and $p_i=1/4$, we have $q_j=\frac14$, so the
  entropy $H(Y)$ attains its maximum, i. e. $H(Y)=\log 4$. Since
  $I(X:Y)=H(Y)-H(Y|X)$, the Proposition remains proved.
\end{proof}

We notice that for the $2$-dimensional SIC POVM
$\Pi=\{\frac12\ket{\pi_j}\bra{\pi_j}\}_{j=0}^3$, the maximally informative
ensemble $V=\{\frac14, \ket{\psi_i}\}_{i=0}^3$ in Prop. \ref{thm:tetra} is
such that any state $\ket{\psi_i}$ is orthogonal to one state $\ket{\pi_j}$.
Since the states of $V$ lie on the vertexes of a tetrahedron, this ensemble is
usually called {\em antitetrahedron}. The accessible information of the
ensemble which enjoys the same symmetry as $\Pi$ has been proven in
\cite{Dav78} to be $\log4/3$. We want to comment that generally SIC POVMs have
low informational power, as it happens for overcomplete measurements: for
informational completeness one must pay the price of low informational power.

\section{Conclusions}\label{sec:conclusions}

In this work we introduced the informational power of a quantum measurement as
the maximum amount of classical information that the POVM can extract from any
ensemble of states. We showed that it is the classical capacity of a
quantum-classical channel and proved additivity. We restated the problem of
maximizing the informational power of a POVM as the problem of maximizing the
accessible information of a suitable ensemble, and provided a bound on the
minimal number of states of a maximally informative ensemble. Then we provided
a numerical algorithm to find a maximally informative ensemble for a given
POVM. Finally, we classified some POVMs according to their informational
power, namely POVMs with commuting elements, real-symmetric and
mirror-symmetric POVMs.

The presented results have obvious practical relevance in several contexts,
such as the communication of classical information over quantum channels and
the storage and retrieval of information from quantum memories.

{\it Note added in the proof.} After the submission of this work, two related
manuscripts appeared on arXiv \cite{Ore11,Hol11}. In particular, Holevo
\cite{Hol11} studied the informational power in the relevant
infinite-dimensional case.

\section*{Acknowledgments}

We thank Francesco Buscemi, Michael Hall, and Jon Tyson for useful
suggestions. This work was supported by the Italian Ministry of Education
through PRIN 2008 and the European Community through the COQUIT and CORNER
projects.

\end{document}